\documentclass[letterpaper, 10 pt, conference]{IEEEconf}

\IEEEoverridecommandlockouts 

\usepackage{graphics} 
\usepackage{epsfig} 
\usepackage{times} 
\usepackage{amsmath} 
\usepackage{amssymb}  
\usepackage{amsfonts}  
\usepackage{bm} 
\usepackage[hidelinks]{hyperref}

\usepackage{amsthm}
\usepackage{xcolor}
\usepackage[noadjust]{cite}
\usepackage{tabularray}

\DeclareMathOperator{\Tr}{Tr}
\DeclareMathOperator{\diagV}{diagV}
\DeclareMathOperator{\diagM}{diagM}

\usepackage{graphicx}
\usepackage{textcomp}
\usepackage{xcolor}
\usepackage{tikz}
\usepackage{algorithm}
\usepackage{algpseudocode}
\usepackage{caption}
\usepackage{subcaption}

\title{\LARGE \bf

Graph-Based Dynamics and Network Control of a\\ Single Articulated Robotic System

}

\author{Jonathan Lane$^{1}$ and Nak-seung Patrick Hyun$^{1}$
\thanks{$^{1}$School of Electrical and Computer Engineering, Purdue University, West Lafayette, IN 47907, USA, {\tt\small \{jglane,nhyun\}@purdue.edu}}%
}

\begin{document}

\maketitle
\thispagestyle{empty}
\pagestyle{empty}

\theoremstyle{plain}
\newtheorem{theorem}{Theorem}
\newtheorem{corollary}{Corollary}
\newtheorem{prop}{Proposition}
\newtheorem{lemma}{Lemma}
\newtheorem{assumption}{Assumption}

\theoremstyle{definition}
\newtheorem{definition}{Definition}
\newtheorem{remark}{Remark}
\newtheorem{prob}{Problem}

\begin{abstract}
Extensive research on graph-based dynamics and control of multi-agent systems has successfully demonstrated control of robotic swarms, where each robot is perceived as an independent agent virtually connected by a network topology. The strong advantage of the network control structure lies in the decentralized nature of the control action, which only requires the knowledge of virtually connected agents. In this paper, we seek to expand the ideas of virtual network constraints to physical constraints on a class of tree-structured robots which we denote as single articulated robotic (SAR) systems. In our proposed framework, each link can be viewed as an agent, and each holonomic constraint connecting links serves as an edge. By following the first principles of Lagrangian dynamics, we derive a consensus-like matrix-differential equation with weighted graph and edge Laplacians for the dynamics of a SAR system. The sufficient condition for the holonomic constraint forces becoming independent to the control inputs is derived. This condition leads to a decentralized leader-follower network control framework for regulating the relative configuration of the robot. Simulation results demonstrate the effectiveness of the proposed control method.

\end{abstract}


\section{Introduction}
Network control theory has been used extensively to model the dynamics and derive control methods for multi-agent systems across many fields of study, such as the formation control of a drone swarm \cite{anderson2008uav}. Furthermore, the formation control methods for multi-agent systems are decentralized in the sense that each agent only uses state information of neighboring agents to decide its control policy. 
For example, in the leader-follower approach to formation control, a leader agent indirectly guides the followers to achieve a cooperative goal (similar to how ducklings follow each other and collectively follow their mother). 
This method can be used to control a swarm of mobile robots or UAV's to a rigid formation governed by a set of desired relative distances between agents \cite{desai1998controlling}. In these traditional applications of multi-agent robotics, the network describes \textit{virtual} constraints acting on each agent rather than \textit{physical} constraints (i.e. holonomic constraints which reduce the degrees of freedom).
In this paper, we study a new perspective on modeling the dynamics of a single articulated robotic (SAR) system with multiple links and joints as a physically constrained network.
The SAR system on a plane is formally defined in Section~\ref{sec:SAR}.


The Kuramoto model applied to a system of coupled metronomes is an example of a network whose connections are physical \cite{boda2013rhythm}. The model represents the phase dynamics of each metronome based on its physical interactions with the other metronomes using the \textit{incidence matrix}, which is a description of the directed graph \cite{jadbabaie2004stability}.
However, the phase model is a first-order approximation of the network interactions and is not sufficient to describe the full nonlinear dynamics of coupled metronomes \cite{ulrichs2009synchronization}.
Conversely, many articulated robots contain interconnections between neighboring links which indicate an underlying graph. For example, DRAGON is an open chain modular flyer that consists of four rigid links with multidirectional thrusters connected by universal joints \cite{zhao2018dragon}. However, the dynamics presented in \cite{zhao2018dragon} were not derived using a network structure. The physical constraints of open-chain robots can provide insight for an underlying graph structure and control from a network perspective, leading to a generalization to more complicated open chains.



Traditionally, graph-based representations have been used to derive the dynamics of generalized robotic mechanisms. The vector-network method is a procedure to generate the equations of motion for a dynamic system from a graph representation of the system using only the \textit{incidence matrix} of the graph
and a set of equations that describe the mechanics of each link \cite{andrews1988general} \cite{mcphee1996use}.
Similarly, the dynamics equations of a tree-structured robot can be derived using the incidence matrix of a graph where the nodes represent mass elements and the edges represent joints \cite{wittenburg1994topological}. However, the constraint forces on the joints are expressed as being independent to external forces, which neglects the general structural dependencies of the control forces based on the graph properties in \cite{wittenburg1994topological}. 
Furthermore, a recursive formulation of the dynamics for open chain manipulators can be derived by considering the net wrench acting on each body as the sum of wrenches from bodies down the chain \cite{park1995lie}. However, the recursive algorithm does not explicitly demonstrate properties of the graph (i.e. incidence matrix, graph Laplacian). 
The above methods have been used extensively to derive the equations of motion of dynamical systems with an algorithm and are not focused on control, especially through the lens of distributed control within a single connected articulated robot.
In this paper, we derive graph-based dynamics for SAR systems which interpret the network coupling between coordinates and forces acting on each joint with respect to holonomic constraints.

The key concept of representing the dynamics of SAR systems purely based on an underlying graph structure lies in the choice of a non-minimal representation of the system \cite{andrews1988general} \cite{wittenburg1994topological}.
The dynamics of a generalized networked mass-spring-damper system has been derived on the spatial coordinates of the masses as a port-Hamiltonian system using a \textit{weighted graph Laplacian} (defined in Section~\ref{sec:prelim}) \cite{van2013port}. Additionally, the consensus equation for the underlying edge-weighted graph of a mass-damper system is shown to be equivalent to the system's generalized momentum equation. The choice of coordinates is similar to those in this paper, however the holonomic constrained Hamiltonian is not considered because the mass elements are not rigidly linked, leading to a minimal representation of the configuration space.
On the other hand, expressing the coupling within the dynamics of a SAR system with a minimal representation can easily break the underlying network (see the motivating example in Section~\ref{sec:motivatin_example}).
By over-parameterizing the configuration space of a SAR system with the inertial coordinates of the centers of masses of each link, 
we can define a graph where the mass elements represent nodes and the relative positions between two nodes represent edges subject to holonomic constraints (i.e. length constraints of rigid links).
In this paper, we propose a formulation of the dynamics equations of a tree-structured robotic system as a matrix-differential equation using a weighted graph Laplacian and weighted edge Laplacian (see definitions in Section~\ref{sec:prelim}).
Along with the forms of the dynamics equations, we present a method to control the relative coordinates between the links to a set of desired relative coordinates similar to a leader-follower controller used in multi-agent formation control literature 
\cite{dai2020adaptive} \cite{pan2021multilayer}.


\section{Preliminaries}
\label{sec:prelim}
\subsection{Graph theory}
In this section, we will define some preliminary concepts of graph theory needed for this paper. For a more detailed introduction to graph theory see \cite{mesbahi2010graph}. A \textit{directed graph}, or \textit{digraph}, denoted $\mathcal{G}:=(V_\mathcal{G}, E_\mathcal{G})$ is defined as the pair of a node set $V_\mathcal{G}$ and an edge set $E_\mathcal{G}$. The node set is a set of $n$ nodes $V_\mathcal{G}:=\{v_1,\hdots,v_n\}$ where $v_i$ represents the $i$th node for $i\in\{1,\hdots,n\}$. The edge set $E_\mathcal{G}\subset V_\mathcal{G}\times V_\mathcal{G}$ is a binary relation with $k$ elements such that if $e_j\in E_\mathcal{G}$, there exists $v_l,v_m\in N_{\mathcal{G}}$ such that $e_j=(v_l,v_m)$ is the $j$th edge in $\mathcal{G}$. For an edge $(v_l,v_m)\in E_\mathcal{G}$, $v_l$ is called the tail of the edge and $v_m$ is called the head.
When an edge exists between two nodes, the nodes are called \textit{adjacent} to each other and \textit{incident} to that edge. 

\begin{definition}[Definition 3.7 in \cite{mesbahi2010graph}]
A digraph $\mathcal{G}$ is an \textit{arborescence} (rooted out-branching tree) if it does not contain a directed cycle (closed loop) and contains a node $v_r$ called the \textit{root} such that for every $i\neq r$, there is exactly one directed path from $v_r$ to $v_i$.
\end{definition}

\begin{lemma} \label{lem:k=n-1}
An arborescence with $n$ nodes has $n-1$ edges.
\end{lemma}
\begin{proof}
See Theorem 4.2 in \cite{dharwadker2011graph}.
\end{proof}
Lemma~\ref{lem:k=n-1} infers that every node in an arborescence is the head of exactly one edge except the root node.
\begin{definition} \label{def:D}
The \textit{incidence matrix} $D(\mathcal{G})\in\mathbb{R}^{n\times k}$ of a digraph $\mathcal{G}$ is defined as follows. Let $d_{ij}$ be an element in the $i$th row and $j$th column of $D(\mathcal{G})$ representing the edge $e_j=(v_l,v_m)$ for some $j\in\{1,\hdots,k\}$ and $l,m\in\{1,\hdots,n\}$. Then
\begin{equation*}
d_{ij}=\begin{cases}
    1 & \text{if }i=m\\
    -1 & \text{if }i=l\\
    0 & \text{otherwise}.
\end{cases}
\end{equation*}
We denote $d_j\in\mathbb{R}^n$ as the $j$th column of $D(\mathcal{G})$ associated with the $j$th edge of $\mathcal{G}$ where $D(\mathcal{G})=[d_1,\hdots,d_k]$.
\end{definition}
The incidence matrix satisfies the following lemma.
\begin{lemma} \label{lem:1_in_null}
The incidence matrix has $\mathbf{1}_n\in\mathcal{N}(D(\mathcal{G})^\top )$
where $\mathbf{1}_n\in\mathbb{R}^n$ is a vector where every element is $1$ \cite{mesbahi2010graph}.
\end{lemma}

\begin{proof}
It follows from Definition~\ref{def:D} that the sum of the elements of $d_j$ is zero for all $j\in\{1,\hdots,k\}$. 
\end{proof}

The following lemma holds if $\mathcal{G}$ is an arborescence.
\begin{lemma} \label{lem:D_rank}
If $\mathcal{G}$ is an arborescence, $D(\mathcal{G})$ is full rank.
\end{lemma}

\begin{proof}
See Lemma 2.2 in \cite{bapat2010graphs}.
\end{proof}

Given an arborescence $\mathcal{G}$, we can cut one edge to divide it into two arborescences. 
\begin{definition}
\label{def:tail_comp}
    Let $e_j=(v_l,v_m)$ represent the $j$th edge of an arborescence $\mathcal{G}$ and $\mathcal{G}\backslash e_j$ represent the graph without $e_j$. We denote the sub-tree of $\mathcal{G}$ containing the tail of the cut edge, $v_l$, as the \textit{tail component} of $\mathcal{G}\backslash e_j$. Similarly, we denote the sub-tree of $\mathcal{G}$ containing the head $v_m$, as the \textit{head component} of $\mathcal{G}\backslash e_j$.
\end{definition}

Another representation of the graph $\mathcal{G}$ is the \textit{graph Laplacian}, which is used in numerous literatures on multi-agent systems \cite{mesbahi2010graph}. 
\begin{definition}
The \textit{edge-weighted graph Laplacian} $L_w(\mathcal{G})\in\mathbb{R}^{n\times n}$ of the undirected version of a digraph $\mathcal{G}$ is defined as $L_w(\mathcal{G}):=D(\mathcal{G})W_eD(\mathcal{G})^\top $ where $W_e\in\mathbb{R}^{k\times k}$ is a diagonal matrix of weights assigned to the edges.
\end{definition}


A \textit{node-weighted edge Laplacian} can be defined as a direct analogy to \textit{edge-weighted graph Laplacian} and describes the adjacency between pairs of edges: 
\begin{definition}
\label{def:weighted_edge_laplacian}The \textit{node-weighted edge Laplacian} $L_e(\mathcal{G})\in\mathbb{R}^{k\times k}$ of the undirected version of a digraph $\mathcal{G}$ is defined as $L_e(\mathcal{G}):=D(\mathcal{G})^\top W_nD(\mathcal{G})$ where $W_n\in\mathbb{R}^{n\times n}$ is a diagonal matrix of weights assigned to the nodes.
\end{definition}

\subsection{Nomenclature}
The notation $\odot$ denotes the Hadamard (elementwise) product of two matrices with the same dimensions. Let $\mathbf{e}_j$ be an elementary vector of appropriate dimension with $1$ in the $j$th position. Let the following be an arbitrary vector and matrix for the next two definitions: $v=[v_1,\hdots,v_N]^\top\in\mathbb{R}^N,~A=(a_{ij})\in\mathbb{R}^{N\times N}$ for some natural number $N$.

\begin{definition}
The function, $\diagM:\mathbb{R}^{N}\rightarrow\mathbb{R}^{N\times N}$, forms a diagonal matrix from the elements of the input vector, and $\diagV:\mathbb{R}^{N\times N}\rightarrow\mathbb{R}^N$ extracts the diagonal of a square matrix as a column vector:
$$\diagM(v)=\begin{bmatrix}
    v_{1}\\
    & \ddots\\
    & & v_{N}
\end{bmatrix},~\diagV(A)=\begin{bmatrix}
    a_{11}\\
    \vdots\\
    a_{NN}
\end{bmatrix}.$$
\end{definition}


\section{Single Articulated Robotic (SAR) System } \label{sec:SAR}
In this section we will provide a definition of a class of SAR systems which we will use in this paper. The SAR system consists of a set of $n$ numbered particles on a plane, where $m_i$ denotes the mass of the $i$th particle for $i\in\{1,\hdots,n\}$.
Also, $n-1$ separate pairs of masses are connected together by massless rods of length $\ell_j$ for $j\in\{1,\hdots,n-1\}$ to form a rigid body linkage. 
Also assume the articulated body is open chain, meaning the rods do not form any closed loops which would restrict the motion of any joints. Let's first look at an example of a SAR system.

\subsection{Motivating example}
\label{sec:motivatin_example}
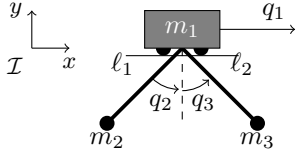
\begin{figure}[htbp]
    \centering
    \begin{tikzpicture}
        \fill[black] (-1, -1) circle (1mm) node[anchor=north] {$m_2$};
        \fill[black] (1, -1) circle (1mm) node[anchor=north] {$m_3$};
        \draw[line width=0.5mm] (0, 0) -- node[anchor=south east] {$\ell_1$} (-1, -1);
        \draw[line width=0.5mm] (0, 0) -- node[anchor=south west] {$\ell_2$} (1, -1);

        \fill (-0.25, 0) circle (0.1);
        \fill (0.25, 0) circle (0.1);
        \filldraw[draw=black, fill=gray] (-0.5, 0) rectangle node[color=white] {$m_1$} (0.5, 0.5);
        \draw (-0.75, -0.1) -- (0.75, -0.1);

        \draw[->] (0.5, 0.25) -- (1.5, 0.25) node[anchor=south east] {$q_1$};
        \draw[dashed] (0, 0) -- node[anchor=north west] {$q_3$} node[anchor=north east] {$q_2$} (0, -1);
        \draw[->] (-0.4, -0.4) arc (225:270:5mm);
        \draw[->] (0, -0.5657) arc (270:315:5mm);

        \draw[->] (-2, 0) node[anchor=north east] {\small $\mathcal{I}$} -- (-1.5, 0) node[anchor=north] {\small $x$};
        \draw[->] (-2, 0) -- (-2, 0.5) node[anchor=east] {\small $y$};
    \end{tikzpicture}
    \caption{Two pendulums on a cart.}
    \label{fig:pendulums}
\end{figure}

Consider the planar system given by Fig.~\ref{fig:pendulums}. The system consists of a cart with mass $m_1$ resting on a platform, and two pendulums of length $\ell_1$ and $\ell_2$ with masses $m_2$ and $m_3$ connected to the cart with revolute joints. Let the generalized coordinates be defined as pictured where $q:=[q_1,q_2,q_3]^\top \in\mathbb{R}^3$, which is a minimal representation of the configuration of the system represented in the inertial frame $\mathcal{I}$.

The generalized mass matrix with respect to the chosen coordinates is denoted as $M(q)\in\mathbb{R}^{3\times 3}$, the Coriolis matrix is denoted as $C(q,\dot{q})\in\mathbb{R}^{3\times 3}$, and the gravitational force and non-conservative forces (e.g. damping on each joint) is denoted as $N(q,\dot{q})\in\mathbb{R}^3$, as defined in \cite{murray2017mathematical}. Then, the dynamics of the system can be expressed using the Euler-Lagrange equation
$\ddot{q} = -M(q)^{-1}(C(q,\dot{q})\dot{q} + N(q,\dot{q}))$.
The mass matrix for this example is written as
$$M(q)=\begin{bmatrix}
    m_1+m_2+m_3 & m_2\ell_1\cos(q_2) & m_3\ell_2\cos(q_3)\\
    m_2\ell_1\cos(q_2) & m_2\ell_1^2 & 0\\
    m_3\ell_2\cos(q_3) & 0 & m_3\ell_2^2
\end{bmatrix}.$$
Notice the sparse structure of the mass matrix, 
which infers that the internal dependencies between each pair of coordinates describes a network of how their accelerations are coupled (zeros in the $(2,3)$ and $(3,2)$ positions indicate the independence between angles $q_2$ and $q_3$).
However, the actual forces 
are coupled according to the inverse of the mass matrix
\begin{equation*}
M(q)^{-1}=f(q)\begin{bmatrix}
    1 & -\frac{\cos(q_2)}{\ell_1} & -\frac{\cos(q_3)}{\ell_2}\\
    -\frac{\cos(q_2)}{\ell_1} &
    M_{22}(q_3) &
    M_{23}(q_2,q_3)\\
    -\frac{\cos(q_3)}{\ell_2} & M_{23}(q_2,q_3)
     & M_{33}(q_2)
\end{bmatrix},
\end{equation*}
for some state dependent expressions $M_{22}(q_3),~M_{23}(q_2, q_3),$ $~M_{33}(q_2)$, and non-zero $f(q)$.

The main reason for the loss of sparsity after taking the inverse of the mass matrix comes from the choice of a minimal set of coordinates to represent the dynamics. 
Therefore, we consider an \textit{over-parameterization} of the configuration space,  
$r_1, r_2, r_3\in\mathbb{R}^2$ where $r_i$ represents the $(x_i,y_i)$ coordinate of the $i$th particle in the inertial frame $\mathcal{I}$. The new structure of generalized coordinates can now be defined as a matrix, not a vector, $Q:=[r_1,r_2,r_3]^\top\in\mathbb{R}^{3\times 2}$.
Since a system model based on $Q$ has more parameters than degrees of freedom, there exist constraint forces due to holonomic constraints representing the fixed rods and restricted motion of the cart. 

On the other hand, the new generalized mass matrix is now diagonal, so its inverse is also diagonal. By using constrained Lagrangian mechanics, we can express the dynamics of this over-parameterized
matrix differential equation as
$$\ddot{Q}=-diagM(\overline{m})^{-1}\Gamma(D(\mathcal{G}), Q, \dot{Q})-\textbf{1}_3\begin{bmatrix}
    0 & g
\end{bmatrix}$$
where $\overline{m}=[m_1,m_2,m_3]^\top$, $\textbf{1}_3=[1,1,1]^\top$, $g\in\mathbb{R}$ is the acceleration due to gravity, and $\Gamma(D(\mathcal{G}), Q, \dot{Q})\in\mathbb{R}^{3\times 2}$ is a matrix representation of the holonomic constraint forces as derived in \cite{murray2017mathematical}. Furthermore, we discover that the algebraic properties encoded in an incidence matrix $D(\mathcal{G})$ representing the connections between masses in Fig.~\ref{fig:pendulums} appear within $\Gamma$, which will be formally defined in Section~\ref{sec:graph}. 

In the following section, we generalize the dynamics of a SAR system with an \textit{arborescence} 
graph based on the interconnections between particles. 

\section{Graph Based Dynamics for a SAR System} \label{sec:graph}
The choice of coordinates for SAR systems make it natural to construct a graph from the structure of the system and the states that describe its configuration. Consider a SAR system model as defined in Section~\ref{sec:SAR} and an arborescence $\mathcal{G}:=(V_\mathcal{G}, E_\mathcal{G})$. Let $v_1$ be the root node in a node set $V_\mathcal{G}=\{v_1,\hdots,v_n\}$. An edge $e_j=(v_l,v_m)$ is in $E_\mathcal{G}$ if two masses, $m_l$ and $m_m$, are connected by a rigid rod of length $\ell_j$ in the SAR system for each $j\in\{1,\hdots,n-1\}$. We call $\mathcal{G}$ the ``underlying graph'' of the SAR system. 

\subsection{Matrix generalized coordinates}
For the $i$th mass in the system, we define its position as $r_i\in\mathbb{R}^2$ for $i\in\{1,\hdots,n\}$ whose values are the inertial coordinates of $m_i$ in the $\mathcal{I}$ frame. Then we can define a matrix of generalized coordinates $Q\in\mathbb{R}^{n\times 2}$ as
$$Q:=[r_1,\hdots,r_n]^\top.$$
We will also refer to $Q$ as the ``node coordinates'' of the system, where the $i$th row of $Q$, $r_i$, corresponds to the position state of $v_i\in V_\mathcal{G}$.


A collection of $n-1$ edge coordinates $r_{ej}\in\mathbb{R}^2$ for $j\in\{1,\hdots,n-1\}$ and an edge coordinates matrix $Q_e\in\mathbb{R}^{(n-1)\times2}$ can be defined with $D(\mathcal{G})$ as follows:
\begin{equation}
\label{eqn:Qe}
    Q_e:=[r_{e1},\hdots,r_{e(n-1)}]^\top =D(\mathcal{G})^\top Q.
\end{equation}
In fact, each $r_{ej}$ is the vector displacement  between masses, represented in $\mathcal{I}$, with a distance constraint $\|r_{ej}\|=\ell_j$.

\subsection{Holonomic constraints}
Let $h_j:\mathbb{R}^{n\times 2}\rightarrow\mathbb{R}$ be a holonomic constraint due to the $j$th rigid rod distance constraint
\begin{equation}
\label{eqn:hol}
h_j(Q(t)):=\tfrac{1}{2}\|Q(t)^\top d_j\|^2-\tfrac{1}{2}\ell_j^2=0
\end{equation}
where $d_j$ is defined in Definition~\ref{def:D}. Observe that all constraints are independent given there are no closed loops in SAR systems. Since the holonomic constraint holds for all $t$, the time derivative of $h_j$ is zero for all $t\in\mathbb{R}_{\geq 0}$. The time derivative of $h_j(Q(t))$ can be written as 
$$\frac{dh_j}{dt}=\Tr\left(\frac{\partial h_j}{\partial Q}\dot{Q}\right)=0$$
where $\Tr(\cdot)$ represents the trace operator acting on a square matrix.
Let $A_j:\mathbb{R}^{n\times 2}\rightarrow\mathbb{R}^{2\times n}$ represent the Jacobian of $h_j$ with respect to $Q$ given by
\begin{equation}
\label{eqn:Aj}
A_j(Q):=\frac{\partial h_j}{\partial Q}
    =\tfrac{1}{2}Q^\top (d_jd_j^\top +d_jd_j^\top )\\
    =Q^\top d_jd_j^\top.
\end{equation}
Therefore, the velocity constraint can be written as
\begin{equation} \label{eq:vel_const}
\Tr\left(A_j(Q)\dot{Q}\right)=\Tr\left(Q^\top d_jd_j^\top \dot{Q}\right)=d_j^\top \dot{Q}Q^\top d_j=0.
\end{equation}
Observe that the velocity constraint in \eqref{eq:vel_const} can be interpreted as a constraint in Pfaffian form after vectorizing $Q$. 

\subsection{External forces} \label{sec:forces}
Suppose there exist independent external forces acting on each particle of the SAR system. Let 
$f_i\in\mathbb{R}^2$ be the vector force acting on $i$th particle of the SAR system represented in the $\mathcal{I}$ frame for each $i\in\{1,\hdots,n\}$. Then we can define a matrix of generalized forces $F\in\mathbb{R}^{n\times 2}$ as
\begin{equation}
\label{eq:control}
    F:=[f_1,\hdots,f_n]^\top.
\end{equation}
In the next section, we will consider these forces as the control parameters of a SAR system.

\subsection{Dynamics of the node coordinates}
Let the mass matrix be $M:=\diagM([m_1,\hdots,m_n]^\top )\in\mathbb{R}^{n\times n}$, which is a constant matrix (state independent). The generalized gravity acting on the node coordinates can be defined in matrix form as
$G:=[\mathbf{0}_n\;g\mathbf{1}_n]\in\mathbb{R}^{n\times2}$ where 
$\mathbf{0}_n,\mathbf{1}_n\in\mathbb{R}^n$ are vectors of all $0$ and all $1$ respectively.

\begin{theorem}
Suppose $\mathcal{G}$ is an arborescence with $n$ nodes and $n-1$ edges representing a SAR system on a plane. Let $Q\in\mathbb{R}^{n\times 2}$ represent the generalized coordinates $Q:=[r_1,\hdots,r_n]^\top $ where $r_i\in\mathbb{R}^2$ for $i\in\{1,\hdots,n\}$, $M\in\mathbb{R}^{n\times n}$ represent the constant mass matrix, $G\in\mathbb{R}^{n\times2}$ represent the generalized gravity matrix, and $F\in\mathbb{R}^{n\times 2}$ represent the matrix of generalized forces. Then, there exists $\lambda=[\lambda_1,\hdots,\lambda_{n-1}]^\top \in\mathbb{R}^{n-1}$ such that the constraint forces $\Gamma\in\mathbb{R}^{n\times 2}$ can be written as
\begin{equation}
\label{eqn:Gamma_thm}
    \Gamma = L_w(\mathcal{G})Q
\end{equation} 
and the matrix differential dynamics equation can be written as 
\begin{equation} \label{eq:consensus}
    \ddot{Q}=-M^{-1}\Gamma+(M^{-1}F-G)
\end{equation}
where $L_w(\mathcal{G})=D(\mathcal{G})\Lambda D(\mathcal{G})^\top $ and $\Lambda=\diagM(\lambda)\in\mathbb{R}^{(n-1)\times(n-1)}$, and $\lambda$ can be written using algebraic operations on $D(\mathcal{G}),Q_e,\dot{Q}_e$, and $F$.


\end{theorem}

\begin{proof}
First, let $\lambda_i$ represent the Lagrange multiplier with respect to the holonomic constraint $h_i$ in \eqref{eqn:hol} for $i\in\{1,\hdots, n-1\}$, and let
$\Gamma\in\mathbb{R}^{n\times 2}$ be a matrix of constraint forces where the $i$th row represents the force due to constraints on the $i$th node as a vector in $\mathbb{R}^2$. The constraint forces are linear combinations of the Jacobians $A_j(Q)^\top$ defined in \eqref{eqn:Aj} with weights $\lambda_j$ as follows:
\begin{equation} \label{eq:Gamma}
    \Gamma:= \sum_{j=1}^{n-1} \lambda_jA_j(Q)^\top
    = \sum_{j=1}^{n-1} d_j\lambda_jd_j^\top Q
    = D(\mathcal{G})\Lambda D(\mathcal{G})^\top Q.
\end{equation}
Observe that $\Gamma$ is equal to \eqref{eqn:Gamma_thm}. Then, the generalized Euler-Lagrange equation for the choice of coordinates $Q$ is written as the matrix sum of forces acting on each mass given by
\begin{equation} \label{eq:EL}
    M\ddot{Q}+MG+\Gamma=F.
\end{equation}
By left multiplying \eqref{eq:EL} by $M^{-1}$ and rearranging terms, we obtain the matrix differential equation as stated in \eqref{eq:consensus}.

Next, we will find a closed-form expression for $\lambda$, starting from the Pfaffian equations \eqref{eq:vel_const}. Since each velocity constraint holds for all time, their time derivatives are also zero for all time and we can write
\begin{equation*}
    \frac{d}{dt}\left[d_j^\top \dot{Q}Q^\top d_j\right]
=d_j^\top \ddot{Q}Q^\top d_j+\|\dot{Q}^\top d_j\|^2=0.
\end{equation*}
Now, by replacing $\ddot{Q}$ in the above equation with \eqref{eq:consensus}, we get 
\begin{equation*}
d_j^\top (-G+M^{-1}(F-D(\mathcal{G})\Lambda D(\mathcal{G})^\top Q))r_{ej}+\|\dot{r}_{ej}\|^2=0.
\end{equation*}
Note that $d_j^\top G=0$ because the range space of $G$ is equal to  $\mathrm{span}\{\mathbf{1}_n\}$. Therefore, the
equation can be simplified to
\begin{equation}
\label{eqn:lambda_step1}
    d_j^\top M^{-1}D(\mathcal{G})\Lambda Q_er_{ej}=d_j^\top M^{-1}Fr_{ej}+\|\dot{r}_{ej}\|^2.
\end{equation}
Observe that since $\Tr(AB)=\Tr(BA)$ for matrices $A$ and $B$ with appropriate dimensions, \eqref{eqn:lambda_step1} can be rewritten as
\begin{equation*}
\Tr\left(Q_er_{ej}d_j^\top M^{-1}D(\mathcal{G})\Lambda\right)=d_j^\top M^{-1}Fr_{ej}+\|\dot{r}_{ej}\|^2.
\end{equation*}
Notice the left side of the equation is the trace of the outer product of the vectors $Q_er_{ej}$ and $D(\mathcal{G})^\top M^{-1}d_j$ times the diagonal matrix $\Lambda$. This is equivalent to the Hadamard (elementwise) product of $D(\mathcal{G})^\top M^{-1}d_j$ and $Q_er_{ej}$, inner product with $\lambda$, which leads to the following equation with simpler notation:
\begin{equation*}
\left(D(\mathcal{G})^\top M^{-1}d_j\odot Q_er_{ej}\right)^\top \lambda=d_j^\top M^{-1}Fr_{ej}+\|\dot{r}_{ej}\|^2.
\end{equation*}
Now, stack these equations vertically for $j\in\{1,\hdots,n-1\}$ to form the equation
$$\left(L_e(\mathcal{G})\odot Q_eQ_e^\top \right)\lambda=\diagV\left(D(\mathcal{G})^\top M^{-1}FQ_e^\top \!+\dot{Q}_e\dot{Q}_e^\top \right)$$
where $L_e(\mathcal{G}):=D(\mathcal{G})^\top M^{-1}D(\mathcal{G})$, a node-weighted edge Laplacian of $\mathcal{G}$ according to Definition~\ref{def:weighted_edge_laplacian}. It follows from Lemma \ref{lem:D_rank} and the fact that $\diagV(M^{-1})$ has all nonzero quantities that $L_e(\mathcal{G})$ is invertible. Since it is also symmetric, $L_e(\mathcal{G})$ is positive definite. Next, observe that the diagonal of $Q_eQ_e^\top$ is entirely nonzero since $\|r_{e,j}\|^2=\ell_j^2$ for all $j$. Since the Hadamard product of a positive definite matrix and a positive semidefinite matrix with no zeros on its diagonal is positive definite (\cite{Djokovic1964}, Lemma~1), $L_e(\mathcal{G})\odot Q_eQ_e^\top$ is positive definite, and so invertible. Therefore, the closed-form expression for $\lambda$ can be written as
\begin{equation}
\label{eq:Lambda}
\lambda=\!\left(L_e(\mathcal{G})\!\odot Q_eQ_e^\top \right)^{\text{--}1}\!\diagV\!\left(D(\mathcal{G})^{\!\top} M^{\text{--}1}FQ_e^\top \!+\dot{Q}_e\dot{Q}_e^\top \right)\!.
\end{equation}
Hence, the Lagrange multipliers $\lambda_j$ can be computed with algebraic operations on $D(\mathcal{G})$, edge coordinates $Q_e$, their derivatives $\dot{Q}_e$, and $F$. 
\end{proof}

The proposed dynamics in \eqref{eq:consensus} for a \textit{single} robotic system can be interpreted as
a \textit{networked} system with node weights $1/m_i$ and state-dependent edge weights $\lambda_j$. 
In general, each node weight in a second-order networked system indicates the magnitude by which that node's acceleration is affected by the dynamics of adjacent nodes \cite{yu2009second}. The node weights are the reciprocal masses of the particles in \eqref{eq:consensus}, meaning that particles with larger masses have smaller accelerations due to the network. Each edge weight in a second-order networked system indicates the magnitude by which the nodes incident to the edge affect each other's dynamics. Also, the edge weights are the Lagrange multipliers of the system constraints which are defined as the relative magnitudes of the constraint forces in \eqref{eq:consensus}. That means larger forces exerted between connecting bars will result in greater accelerations on the masses connected to them.

\begin{remark}
The form of the dynamics equation in \eqref{eq:consensus} resembles the classical second-order consensus equation for multi-agent systems presented in \cite{yu2009second} and other relevant literature. The main difference between graph-based network dynamics considered in multi-agent systems literature and the proposed dynamics formulation lies in the fact that the connection is \textit{physical} in SAR systems. 
\end{remark}

\subsection{Dynamics based on edge coordinates}
Notice that the matrix differential dynamics equation in \eqref{eq:consensus} contains a mixture of node and edge configurations. In this section, we propose the \textit{edge dynamics} of the given SAR model, which is only a function of the edge coordinates $Q_e$, similar to the edge Laplacian dynamics presented in \cite{zeng2016convergence}.

\begin{corollary}
\label{cor:edge_dyn}
Suppose that $Q_e=D(\mathcal{G})^\top Q$ represents the edge coordinates of the SAR system. Then the ``edge dynamics'' can be written as
\begin{equation} \label{eq:edge}
    \ddot{Q}_e=-L_e(\mathcal{G})\Lambda Q_e+D(\mathcal{G})^\top M^{-1}F
\end{equation}
where $L_e(\mathcal{G}):=D(\mathcal{G})^\top M^{-1}D(\mathcal{G})$.
\end{corollary}

\begin{proof}
By using the definition of edge coordinates in \eqref{eqn:Qe}, we can compute the second derivative of $Q_e$ from the closed form expression of $Q$ in \eqref{eq:consensus} given by
\begin{equation*}\ddot{Q}_e=D(\mathcal{G})^\top \ddot{Q}=D(\mathcal{G})^\top \left(-M^{-1}(L_w(\mathcal{G})Q+F)-G\right).
\end{equation*}
Note that $D(\mathcal{G})^\top G=0$ since the range space of $G$ is $\mathrm{span}\{\mathbf{1}_n\}$, $Q_e=D(\mathcal{G})^\top Q$, and $L_w(\mathcal{G})=D(\mathcal{G})\Lambda D(\mathcal{G})^\top$. Then, the equation further simplifies to 
\begin{equation*}
    \ddot{Q}_e=-L_e(\mathcal{G})\Lambda Q_e+D(\mathcal{G})^\top M^{-1}F
\end{equation*}
where $L_e(\mathcal{G})$ is the same as defined in Corollary~\ref{cor:edge_dyn}.
\end{proof}
\begin{remark}
\label{rem:force}
    Note that the node dynamics \eqref{eq:consensus} and edge dynamics \eqref{eq:edge} are not in the traditional control affine form as $\Lambda$ depends on the control force $F$. In the next section, we derive the sufficient condition for $F$ such that $\Lambda$ becomes independent to the control $F$.
\end{remark}

\section{Decentralized Network Control for SAR Systems}
\label{sec:control}
In a SAR system, the root node of the arborescence $\mathcal{G}$ can serve as a leader while the other nodes can be seen as followers working towards reaching a harmony based on their interconnections (edges). Since the edge perspective of the dynamics in \eqref{eq:edge} is formulated in such a way where the net control force applied to each edge directly affects its own acceleration plus network terms, we consider the following control problem.
\begin{prob}[Edge coordinate control problem]
\label{prob:1}
Let $r_{ej,d}\in\mathbb{R}^2$ be a desired coordinate for the $j$th edge in a SAR system satisfying the holonomic constraint $\|r_{ej,d}\|=\ell_j$, and let $Q_{e,d}=[r_{e1,d},\hdots,r_{e(n-1),d}]^\top \in\mathbb{R}^{(n-1)\times 2}$ be the desired coordinates for all edges. The goal is to design a controller for $F$ in \eqref{eq:control} such that $Q_e(t)$ approaches $Q_{e,d}$ as $t$ goes to infinity, where the leader (the root node), is controlled independently with some open-loop control $f_l(t)\in\mathbb{R}^2$ at time $t$.
\end{prob}
Note that for a constant desired edge coordinate $Q_{e,d}$, the derivative of $Q_{e,d}$ is zero. Therefore, with $F=\textbf{0}_{n\times 2}$, control $Q_{e,d}$ is an equilibrium of the edge dynamics \eqref{eq:edge}.

\subsection{Network based feedback controller design}

First, observe that the edge dynamics in \eqref{eq:edge} is over-actuated as there are $n-1$ edges and $n$ degrees of control. However as stated in Problem~\ref{prob:1}, the control of the leader node
is independent, and so the proposed SAR system is fully actuated. As stated in Remark~\ref{rem:force}, the vector of Lagrange multipliers $\lambda$ is a function of $F$. The following proposition shows the sufficient condition for $\Lambda$ to be independent to the control and only dependent on the network topology and current edge coordinates and their velocities.
\begin{prop}
\label{prop:suff}
    Let $U\in\mathbb{R}^{(n-1)\times 2}$ be an arbitrary matrix such that the $j$th row of $U$, denoted $u_j^\top\in\mathbb{R}^{1\times 2}$, is orthogonal to the $j$th edge coordinate, $e_j^\top Q_e(t)=r_{ej}^\top(t)$ at time $t$ for all $j\in\{1,\hdots, n-1\}$. Then there exists a control $F\in\mathbb{R}^{n\times 2}$ such that
    \begin{equation} \label{eq:plug}
D(\mathcal{G})^\top M^{-1}F=U
\end{equation}
and $\lambda$ in \eqref{eq:Lambda} is independent from $F$ at time $t$.
\end{prop}
\begin{proof}
First, observe that $D(\mathcal{G})^\top M^{-1}$ with dimension $(n-1)\times n$ is full row rank since $D(\mathcal{G})$ is full column rank and $M$ is nonsingular. Therefore, there exists $F$ satisfying any $U\in\mathbb{R}^{(n-1)\times 2}$. Now, suppose that  $U\in\mathbb{R}^{(n-1)\times 2}$ satisfies the condition stated in the proposition where each row is orthogonal to the matched row in $Q_e(t)$ at time $t$. Note that the control $F$ appears in the $\lambda$ expression in \eqref{eq:Lambda} as \begin{equation}
        \diagV(D(\mathcal{G})^\top M^{-1}FQ_e^\top)=\diagV(UQ_e^\top)=\textbf{0}_{n-1}.
    \end{equation}
    Therefore, $\lambda$ is independent to the $F$ solving \eqref{eq:plug} with each row of $U$ being orthogonal to the matched row of $Q_e(t)$ at time $t$. 
\end{proof}

The following lemma shows that the inverse of $D(\mathcal{G})^\top M^{-1}$ can be directly computed without using the Moore-Penrose pseudo-inverse, but by using the properties of graphs. 
\begin{lemma}
\label{lem:H}
    Suppose that $\mathcal{G}$ is an arborescence and the matrix $H(\mathcal{G})\in\mathbb{R}^{(n-1)\times n}$ is defined such that the $(i,j)$th component is
    \begin{equation}
        h_{ij} = \begin{cases}
            1 & \text{if $v_j$ is in the head component of }\mathcal{G}\backslash e_i\\
            0 & \text{otherwise}
        \end{cases},
    \end{equation}
    where the head component of a graph is defined in Definition~\ref{def:tail_comp}. Then
    \begin{equation}
        H(\mathcal{G})D(\mathcal{G})=I_{n-1}
    \end{equation}
    holds.
\end{lemma}
\begin{proof}
Define $H(\mathcal{G})^*\in\mathbb{R}^{(n-1)\times n}$ such that $H(\mathcal{G})^*:= H(\mathcal{G})-\textbf{1}_{n-1}\textbf{1}_{n}^\top.$
By invoking Lemma~4.15 in \cite{bapat2010graphs}, we have
$$H(\mathcal{G})^*D(\mathcal{G})=I_{n-1}.$$
Since $D(G)^\top\textbf{1}_n=\textbf{0}_{n-1}$ from Lemma~\ref{lem:1_in_null}, 
    \begin{equation}
        H(\mathcal{G})D(\mathcal{G}) = (H(\mathcal{G})^*+\textbf{1}_{n-1}\textbf{1}_n^\top)D(\mathcal{G}) = I_{n-1}
    \end{equation}
    holds.
\end{proof}
Observe that the left inverse $H(\mathcal{G})$ of $D(\mathcal{G})$ is only composed with zeros and ones.  
\begin{corollary}
    If $F(t)$ has the following structure with $U=[u_1,\hdots,u_{n-1}]^\top\in\mathbb{R}^{(n-1)\times 2}$ where $u_j\in\mathbb{R}^2$ is orthogonal to $r_{e,j}$ at time $t$:
    \begin{equation}
    \label{eq:proposed_controller}
        F(t)=M(\alpha\textbf{1}_nf_l(t)^\top+G+H(\mathcal{G})^\top U)
    \end{equation}
    where $\alpha=m_1^{-1}$ is a constant scalar and $G:=[\mathbf{0}_n\;g\mathbf{1}_n]$ is the gravity term, then $\lambda$ is independent of $F$ and $f_1(t)=f_l(t)+m_1g\mathbf{e}_2$, and the edge dynamics can be written as 
    \begin{equation} \label{eq:leader_follower}
    \ddot{Q}_e=-L_e(\mathcal{G})\Lambda Q_e+U.
\end{equation}
\end{corollary}
\begin{proof}
    By substituting the proposed $F$ into \eqref{eq:plug} and invoking Lemma~\ref{lem:1_in_null} and Lemma~\ref{lem:H}, we can see that the sufficient condition in Proposition~\ref{prop:suff} is satisfied. In addition, since $\mathcal{G}$ is an arborescence (rooted out-branching tree), $H(\mathcal{G})\mathbf{e}_1=\textbf{0}_{n-1}$ holds since no head component contains the root node. Therefore, 
\begin{equation*}
    f_1(t) = F(t)^\top\mathbf{e}_1 = m_1(\alpha f_l(t)+g\mathbf{e}_2) = f_l(t) + m_1g\mathbf{e}_2.
\end{equation*}
Hence, the edge dynamics using the proposed control framework will be equal to \eqref{eq:leader_follower}.
\end{proof}

\subsection{Recursively defined follower controller design}

Given an arborescence graph $\mathcal{G}$, and the corresponding control framework in \eqref{eq:proposed_controller}, denote $f_i$ as a \textit{follower controller} if $i\in\{2,\hdots,n\}$ is in the set of non-root nodes. The follower controller can be recursively defined based on the edges incident to each node by separating the rows of \eqref{eq:plug} and using the fact that $f_1=f_l + m_1g\mathbf{e}_2$,
\begin{equation}
    f_i = \bar{f}_i + m_ig\mathbf{e}_2~~\mathrm{where}~~
    \bar{f}_i := m_i\left(\frac{\bar{f}_k}{m_k} + u_j\right)
\end{equation}
and $e_j=(v_k,v_i)\in E_\mathcal{G}$ for all $i\in\{2,\hdots,n\}$.

\subsection{Decentralized follower controller design}
Let the $j$th edge errors be
\begin{equation*}
e_{c,j}:=r_{ej}-r_{ej,d},\quad
e_{v,j}:=\dot{r}_{ej}-\dot{r}_{ej,d}
\end{equation*}
where the subscripts, $c, v$ represent the error of the coordinate and its velocity respectively. 
The sufficient condition for the new control parameters in Proposition~\ref{prop:suff} is to make sure each $u_j$ is orthogonal to $r_{ej}(t)$ at time $t$. Let the time-varying projection matrix, $P_j(t)\in\mathbb{R}^{2\times 2}$, be defined as 
\begin{equation*}
    P_j(t):=I_2-\frac{r_{ej}(t)r_{ej}(t)^\top }{\ell_j^2}
\end{equation*}
which is a projection onto the complement of the space spanned by $r_{ej}(t)$.
Now, choose $u_j$ to be the projection of a simple proportional and differential controller based on the edge tracking error onto that complement space given by
\begin{equation}
    \label{eq:final_control_u}
    u_j=P_j(-k_ce_{c,j}-k_ve_{v,j})
\end{equation}
where $k_c,k_v\in\mathbb{R}$ are positive control gains.
Therefore, the feedback control of the $i$th follower node can be expressed using the proposed controller in \eqref{eq:proposed_controller} and $u_j$ in \eqref{eq:final_control_u} as
\begin{equation}
\label{eq:feedback_control}
    f_i = \bar{f}_i + m_ig\mathbf{e}_2,~
    \bar{f}_i = m_i\left(\frac{\bar{f}_k}{m_k} + P_j(-k_ce_{c,j}-k_ve_{v,j})\right)
\end{equation}
where $v_k$ is the tail node of the edge $e_j$.

\begin{remark}
    The proposed controller in \eqref{eq:feedback_control} is decentralized in the sense that the force applied to each follower node only requires time-dependent knowledge of the force, position, and velocity of the upstream node, in addition to local information.
    

\end{remark}

\subsection{Upperbound of the residual vector field}

The error dynamics of $j$th edge of the SAR system can be expressed as 
\begin{equation*}
\begin{split}
    \dot{e}_{c,j}&=e_{v,j},\\
    \dot{e}_{v,j}&=-k_cP_je_{c,j}-k_vP_je_{v,j}-Q_e^\top \Lambda L_e(\mathcal{G})\mathbf{e}_j.
\end{split}
\end{equation*}
Define $X_j\in\mathbb{R}^{2}$ as $X_j:=-Q_e^\top \Lambda L_e(\mathcal{G})\mathbf{e}_j$. Then there exists a nonzero constant vector $b\in\mathbb{R}^{n-1}$ such that $X_j=-Q_e^\top \Lambda b$
since $L_e(\mathcal{G})\mathbf{e}_j$ is a nonzero constant vector. Also, since $\Lambda$ is a diagonal matrix, we have the following inequalty:
\begin{equation}
    \|X_j\|_2\leq \beta_1\|Q_e^\top\|_2\|\lambda\|_2
\end{equation}
for some $\beta_1>0$ where $\|\cdot\|_2$ is the induced $l_2$ norm. By using inequalities between the Frobenius and $l_2$ norms, and the $l_2$ and $l_1$ norms, a new bound on $\|X_j\|_2$ can be obtained:
\begin{eqnarray*}
    \|X_j\|_2\leq \beta_2\|\lambda\|_1&\leq& \beta_2\|J^{-1}\|_\mathcal{F}\|\diagV(\dot{Q}_e\dot{Q}_e^\top)\|_1
    \\&=&\beta_2\|J^{-1}\|_\mathcal{F}~\Tr(\dot{Q}_e\dot{Q}_e^\top)
\end{eqnarray*}
for some $\beta_2>0$ and $J=L_e(\mathcal{G})\odot Q_eQ_e^\top$ where $\|\cdot\|_\mathcal{F}$ represents the Frobenius norm. Observe that $J$ is positive definite for all time $t$. Let $\sigma(J)>0$ represent the smallest eigenvalue of $J$, then we have 
\begin{equation}
\label{eq:bound}
    \|X_j\|_2\leq \frac{\beta_2}{\sigma(J)}\Tr(\dot{Q}_e\dot{Q}_e^\top)=\frac{\beta_2}{\sigma(J)}\sum_{j=1}^{n-1}\|e_{v,j}\|_2^2
\end{equation}
where the last equality holds since $\dot{Q}_{e,d}=\textbf{0}_{(n-1)\times 2}$ as defined in Problem~\ref{prob:1}.
\begin{remark}
    The bound on $X_j$ holds for all $j\in\{1,\hdots,n-1\}$. Therefore, by constructing the error state for all edges as a single vector, the bound in \eqref{eq:bound} can potentially be used to show the ultimate boundedness of the proposed controller by invoking Theorem 4.18 in \cite{Khalilbook}. On the other hand, observe that the system has $n+1$ degrees of freedom and $2n$ degrees of actuation without the sufficient condition in Proposition~\ref{prop:suff}. Therefore, we can always find a solution to the centralized version of the control $F$ via feedback linearization.
\end{remark}

\section{Results}
\label{sec:results}
In this section we present simulated results from applying the proposed controller \eqref{eq:feedback_control} on a two-link SAR system (Fig.~\ref{fig:2-link_diagram}), and modified feedback controller to track time-varying desired edge trajectory for a five-link SAR system (Fig.~\ref{fig:5-link_diagram}). The resulting plots in Fig.~\ref{fig:2-link_plot} and Fig.~\ref{fig:5-link_plot} are arranged like the coordinates of $Q_e$. 


\subsection{Two-link SAR system}
\begin{figure}[t]
    \centering
    \begin{subfigure}[t]{0.40\linewidth}
        \centering
        \begin{tikzpicture}
            \draw[->, line width=0.5mm] (0, 0) -- (-1, -1);
            \draw[->, line width=0.5mm] (0, 0) -- (1, -1);
            \node at (0.8, -0.3) {$r_{e2}$};
            \node at (-0.8, -0.3) {$r_{e1}$};

            \draw[->, blue, line width=0.3mm] (-1.05, -1.05) -- (-1.55, -0.55) node[anchor=east] {$f_2$};
            \draw[->, blue, line width=0.3mm] (1.05, -1.05) -- (1.55, -0.55) node[anchor=west] {$f_3$};

            \filldraw[fill=white, draw=red, line width=0.6mm] (0, 0) circle (1mm);
            \fill[blue] (-1.05, -1.05) circle (1mm);
            \fill[blue] (1.05, -1.05) circle (1mm);
            \node at (0, 0.3) {$m_1$};
            \node at (-1.05, -1.35) {$m_2$};
            \node at (1.05, -1.35) {$m_3$};
    
            \draw[->] (-2, 0.25) node[anchor=north east] {\small $\mathcal{I}$} -- (-1.5, 0.25) node[anchor=north] {\small $x$};
            \draw[->] (-2, 0.25) -- (-2, 0.75) node[anchor=east] {\small $y$};
        \end{tikzpicture}
        \caption{}
        \label{fig:2-link_diagram}
    \end{subfigure}
    \hfill
    \begin{subfigure}[t]{0.40\linewidth}
        \centering
        \begin{tikzpicture}[node distance={12mm}, thick, main/.style = {draw, circle}]
            \tikzset{edge/.style = {->,> = latex'}};
            \node[main] (1) {$1$};
            \node[main] (2) [below left of=1] {$2$};
            \node[main] (3) [below right of=1] {$3$};
            \draw[->] (1) -- node[anchor=south east] {$1$} (2);
            \draw[->] (1) -- node[anchor=south west] {$2$} (3);
        \end{tikzpicture}
        \caption{}
        \label{fig:2-link_graph}
    \end{subfigure}
    \vfill
    \begin{subfigure}[t]{\linewidth}
        \centering
        \includegraphics[width=\linewidth]{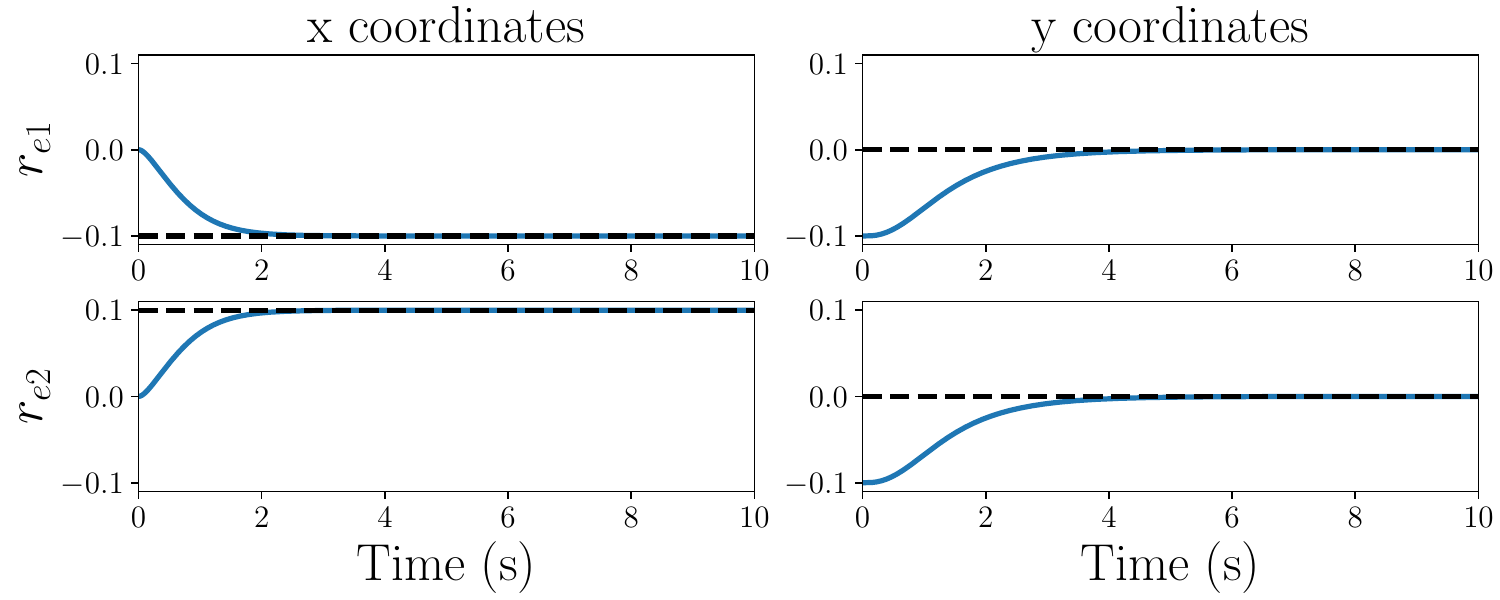}
        \caption{}
        \label{fig:2-link_plot}
    \end{subfigure}
    \caption{Two-link SAR system example. (a) System model with leader (red/white), followers (blue), and applied forces. (b) Arborescence $\mathcal{G}$ diagram. (c) Tracking results for $Q_e$ (blue) to the desired $Q_{e,d}$ (black/dashed).}
\end{figure}
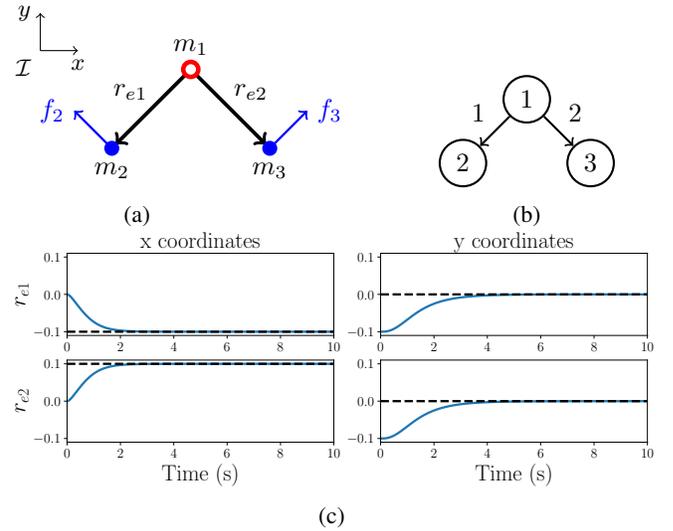


Consider the SAR system shown in Fig.~\ref{fig:2-link_diagram}, which can be viewed as a simplified model of a forward-facing bird in planar space. Each mass can exert a force as defined in Section~\ref{sec:forces} to control both the global position and relative configuration of the system. The underlying directed graph $\mathcal{G}$ is shown in Fig.~\ref{fig:2-link_graph}. 
The system parameters are as follows: $m_1=0.7,~m_2=m_3=0.2,~\ell_1=\ell_2=0.1$.
Control gains are $k_c=k_v=10$. The desired edge coordinates and leader force are chosen as follows: $r_{e1, d} = [-l_1, 0]^\top$ and $r_{e2,d}=[l_2, 0]^\top$ with $f_l(t)=[0,0.5\cos(\pi t)]^\top$.
The decentralized follower controller in \eqref{eq:feedback_control} is applied to $f_2$ and $f_3$, and the results shown in Fig.~\ref{fig:2-link_plot} illustrate that the edge coordinates approach and stay near a constant desired setpoint.
See an animation at (\url{https://youtu.be/2IuMvJKZ8vw}).

Using a graph-based approach, the model can be easily extended to a more detailed representation of a bird's anatomy by adding more links to the network structure without changing the general structure of the dynamics.
Also, the complexity of the control algorithm scales linearly as more links are added due to its distributed nature, as opposed to traditional approaches.



\subsection{Five-link SAR system}

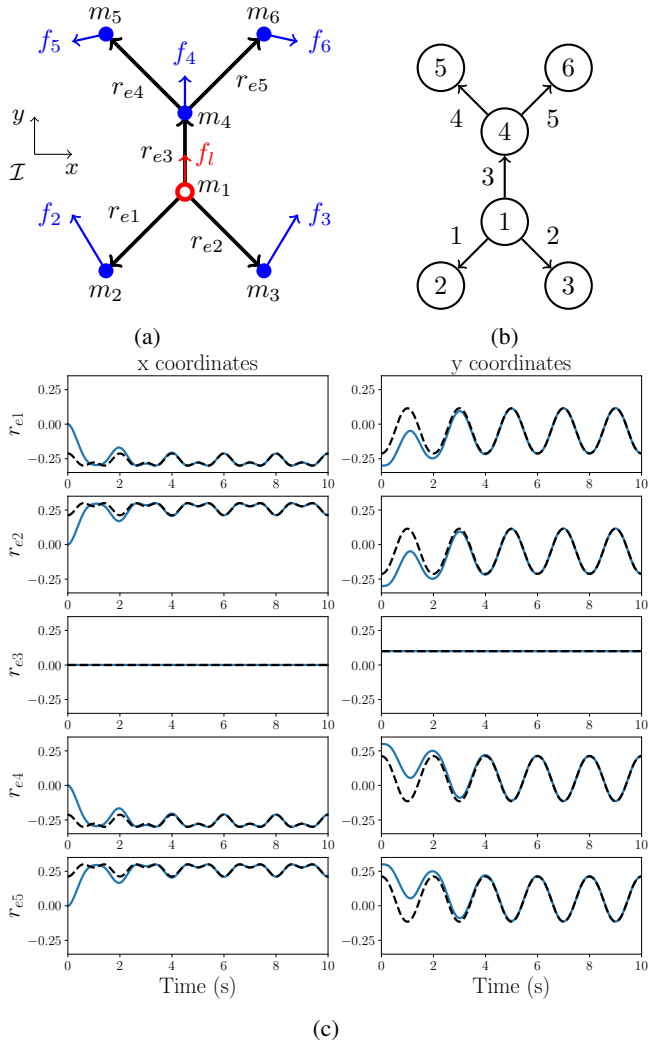
\begin{figure}[htbp]
    \centering
    \begin{subfigure}[t]{0.45\linewidth}
        \centering
        \begin{tikzpicture}
            \draw[->, line width=0.5mm] (0, 1.05) -- (1, 2.05);
            \draw[->, line width=0.5mm] (0, 1.05) -- (-1, 2.05);
            \draw[->, line width=0.5mm] (0, 0) -- (-1, -1);
            \draw[->, line width=0.5mm] (0, 0) -- (1, -1);
            \draw[->, line width=0.5mm] (0, 0) -- node[anchor=east] {$r_{e3}$} (0, 1);
            \node at (-0.8, -0.3) {$r_{e1}$};
            \node at (0.3, -0.7) {$r_{e2}$};
            \node at (-0.75, 1.35) {$r_{e4}$};
            \node at (0.9, 1.45) {$r_{e5}$};
    
            \draw[->, red, line width=0.3mm] (0, 0) -- (0, 0.5) node[anchor=west] {$f_l$};
            \draw[->, blue, line width=0.3mm] (-1.05, -1.05) -- (-1.5, -0.3) node[anchor=east] {$f_2$};
            \draw[->, blue, line width=0.3mm] (1.05, -1.05) -- (1.5, -0.3) node[anchor=west] {$f_3$};
            \draw[->, blue, line width=0.3mm] (0, 1.05) -- (0, 1.55) node[anchor=south] {$f_4$};
            \draw[->, blue, line width=0.3mm] (-1.05, 2.1) -- (-1.5, 2) node[anchor=east] {$f_5$};
            \draw[->, blue, line width=0.3mm] (1.05, 2.1) -- (1.5, 2) node[anchor=west] {$f_6$};
            
            \filldraw[fill=white, draw=red, line width=0.6mm] (0, 0) circle (1mm);
            \fill[blue] (1.05, 2.1) circle (1mm);
            \fill[blue] (-1.05, 2.1) circle (1mm);
            \fill[blue] (-1.05, -1.05) circle (1mm);
            \fill[blue] (1.05, -1.05) circle (1mm);
            \fill[blue] (0, 1.05) circle (1mm);
            \node at (0.4, 0.05) {$m_1$};
            \node at (-1.05, -1.35) {$m_2$};
            \node at (1.05, -1.35) {$m_3$};
            \node at (0.4, 0.95) {$m_4$};
            \node at (-1.05, 2.35) {$m_5$};
            \node at (1.05, 2.35) {$m_6$};
    
            \draw[->] (-2, 0.5) node[anchor=north east] {\small $\mathcal{I}$} -- (-1.5, 0.5) node[anchor=north] {\small $x$};
            \draw[->] (-2, 0.5) -- (-2, 1) node[anchor=east] {\small $y$};
        \end{tikzpicture}
        \caption{}
        \label{fig:5-link_diagram}
    \end{subfigure}
    \hfill
    \begin{subfigure}[t]{0.45\linewidth}
        \centering
        \begin{tikzpicture}[node distance={12mm}, thick, main/.style = {draw, circle}]
            \tikzset{edge/.style = {->,> = latex'}};
            \node[main] (1) {$1$};
            \node[main] (2) [below left of=1] {$2$};
            \node[main] (3) [below right of=1] {$3$};
            \node[main] (4) [above of=1] {$4$};
            \node[main] (5) [above left of=4] {$5$};
            \node[main] (6) [above right of=4] {$6$};
            \draw[->] (1) -- node[anchor=south east] {$1$} (2);
            \draw[->] (1) -- node[anchor=south west] {$2$} (3);
            \draw[->] (1) -- node[anchor=east] {$3$} (4);
            \draw[->] (4) -- node[anchor=north east] {$4$} (5);
            \draw[->] (4) -- node[anchor=north west] {$5$} (6);
        \end{tikzpicture}
        \caption{}
        \label{fig:5-link_graph}
    \end{subfigure}
    \vfill\begin{subfigure}[h]{\linewidth}
        \centering
        \includegraphics[width=\linewidth]{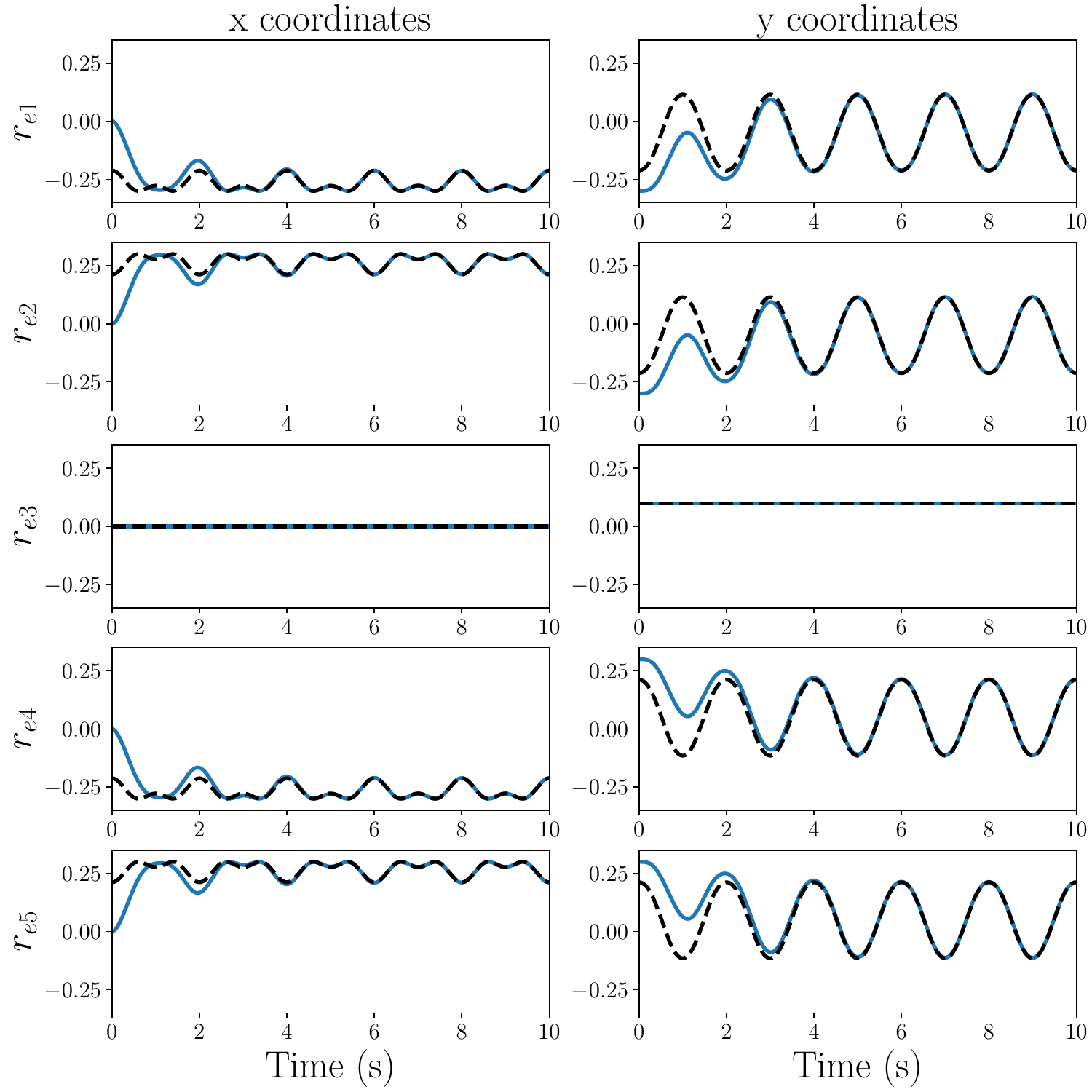}
        \caption{}
        \label{fig:5-link_plot}
    \end{subfigure}
    \caption{Five-link SAR system example. (a) System model with leader (red/white), followers (blue), and applied forces. (b) Arborescence $\mathcal{G}$ diagram. (c) Tracking results. }
\end{figure}
Consider the SAR system shown in Fig.~\ref{fig:5-link_diagram}, and the underlying graph $\mathcal{G}$
shown in Fig.~\ref{fig:5-link_graph}. 
The incidence matrix $D(\mathcal{G})$ and its left inverse $H(\mathcal{G})$ are given by
$$\begin{bmatrix}
    -1 & -1 & -1 & 0 & 0\\
    1 & 0 & 0 & 0 & 0\\
    0 & 1 & 0 & 0 & 0\\
    0 & 0 & 1 & -1 & -1\\
    0 & 0 & 0 & 1 & 0\\
    0 & 0 & 0 & 0 & 1
\end{bmatrix},~\begin{bmatrix}
    0 & 1 & 0 & 0 & 0 & 0\\
    0 & 0 & 1 & 0 & 0 & 0\\
    0 & 0 & 0 & 1 & 1 & 1\\
    0 & 0 & 0 & 0 & 1 & 0\\
    0 & 0 & 0 & 0 & 0 & 1
\end{bmatrix}$$
respectively. The system parameters are defined as$$(m_1,m_2,m_3,m_4,m_5,m_6)=(0.7,0.2,0.2,0.5,0.1,0.1)$$ and $\ell_j=0.3$ for all $j$. 
Control gains are chosen to be $10$ for both $k_c$ and $k_v$,
and $Q_{e,d}(t)$ has a sinusoidal trajectory for the outer links and leader force, which creates a ``flapping'' motion, and a constant trajectory for the center link with $\theta(t) = \frac{3\pi}{16}\cos(\pi t) + \frac{\pi}{16}$ given by
$$Q_{e,d}(t) = \begin{bmatrix}
    -\ell_1\cos\theta(t) & -\ell_1\sin\theta(t)\\
    \ell_2\cos\theta(t) & -\ell_2\sin\theta(t)\\
    0 & \ell_3\\
    -\ell_4\cos\theta(t) & \ell_4\sin\theta(t)\\
    \ell_5\cos\theta(t) & \ell_5\sin\theta(t)
\end{bmatrix}.$$
Note that $\dot{Q}_{e,d}(t),~\ddot{Q}_{e,d}(t)\neq0$, which violates the condition in Problem~\ref{prob:1}. Therefore, we adjust the controller in \eqref{eq:final_control_u} to compensate for the desired acceleration where
\begin{equation}
\label{eq:proposed_future_control}
    u_j=P_j(-k_ce_{c,j}-k_ve_{v,j}+\ddot{r}_{ej,d}).
\end{equation}
The leader node force $f_l$ is chosen to be
$f_l(t) = [0, \sin(2\pi t)]^\top$.
Initial conditions on the edge coordinates are such that each link is aligned vertically, in other words the $x$ component of each $r_{ej}(t)$ is initially zero, and the initial error for $e_{c,3}$ is also zero. The results shown in Fig.~\ref{fig:5-link_plot} illustrate that the edge coordinates approach and stay near a desired trajectory with nonzero velocity and acceleration with \eqref{eq:proposed_future_control} (animation: \url{https://youtu.be/Ba_99zDBZDQ}).
\begin{remark}
The value of the smallest eigenvalue of $J$ during the simulation in Fig.~\ref{fig:5-link_plot} was at a minimum at $t=0$. Thus, the case where all $r_{ej}$'s  are aligned potentially pertains to a lower bound on the smallest eigenvalue of $J$ \cite{fiedler1983note}.
\end{remark} 


\section{Conclusion}
We developed a generalized underlying graph to represent network connections in SAR systems with holonomic constraints and derived the dynamic equations in terms of graph matrices, expressed in both absolute and relative spatial coordinates. Furthermore, we derived a decentralized control strategy to control the relative coordinates towards desired orientations. Simulated results demonstrate the potential stability and effectiveness of the proposed controller. The results in this paper were achieved for SAR systems on a plane, but can be generalized in several ways, such as for robots in spatial space and torque inputs on the joints. The future work of this paper includes theoretical extensions to prove a stability condition of the decentralized controller, generalization to time-varying trajectory tracking, and avoiding inter-component collisions.


\bibliographystyle{IEEEtran}
\bibliography{refs}

\end{document}